\documentclass[11pt]{article}
\usepackage{geometry}
\usepackage{amsthm}
\usepackage{amsfonts}
\usepackage{amsmath}
\usepackage{graphicx}
\usepackage{epstopdf}
\usepackage{lineno}
\usepackage{setspace} 
\usepackage{verbatim}

\usepackage{lastpage}
\usepackage[T1]{fontenc}

\newtheorem{theorem}{Theorem}
\newtheorem*{theorem*}{Theorem}

\newtheorem{corollary}[theorem]{Corollary}

\newtheorem{lemma}[theorem]{Lemma}

\newcommand{\dist}{\mathit{dist}}
\newcommand{\next}{\mathit{Next}}
\newcommand{\len}{\mathit{len}}

\begin{document}

\title{Finding a smallest odd hole in a claw-free graph using global structure}

\author{W.~Sean Kennedy\thanks{School of Computer Science, McGill University, Montreal.} \and  Andrew D.\ King\thanks{Corresponding author: {\tt andrew.d.king@gmail.com}, IEOR Department, Columbia University, New York.  Research supported by an NSERC postdoctoral fellowship.}}

\date{}

\maketitle

\begin{abstract}
A lemma of Fouquet implies that a claw-free graph contains an induced $C_5$, contains no odd hole, or is quasi-line.  In this paper we use this result to give an improved shortest-odd-hole algorithm for claw-free graphs by exploiting the structural relationship between line graphs and quasi-line graphs suggested by Chudnovsky and Seymour's structure theorem for quasi-line graphs.  Our approach involves reducing the problem to that of finding a shortest odd cycle of length $\geq 5$ in a graph.  Our algorithm runs in $O(m^2+n^2\log n)$ time, improving upon Shrem, Stern, and Golumbic's recent $O(nm^2)$ algorithm, which uses a local approach.  The best known recognition algorithms for claw-free graphs run in $O(m^{1.69}) \cap O(n^{3.5})$ time, or $O(m^2) \cap O(n^{3.5})$ without fast matrix multiplication.
\end{abstract}


\section{Background and motivation}

A {\em hole} in a graph is an induced cycle $C_k$ of length $k \geq 4$.  Odd holes are fundamental to the study of perfect graphs \cite{chudnovskyrst06}; although there are polynomial-time algorithms that decide whether or not either a graph or its complement contains an odd hole \cite{cornuejolslv03, chudnovskyclsv05}, no general algorithm for detecting an odd hole in a graph is known.

Odd holes are also fundamental to the study of claw-free graphs, i.e.\ graphs containing no induced copy of $K_{1,3}$.  Every neighbourhood $v$ in a claw-free graph has stability number $\alpha(G[N(v)])\leq 2$.  So if $G[N(v)]$ is perfect then $v$ is bisimplicial (i.e.\ its neighbours can be partitioned into two cliques, i.e.\ $G[N(v)]$ is cobipartite), and if $G[N(v)]$ is imperfect then $G[N(v)]$ contains the complement of an odd hole.  Fouquet proved something stronger:

\begin{lemma}[Fouquet \cite{fouquet93}]\label{lem:fouquet}
Let $G$ be a connected claw-free graph with $\alpha(G)\geq 3$.  Then every vertex of $G$ is bisimplicial or contains an induced $C_5$ in its neighbourhood.
\end{lemma}

It follows that a claw-free graph $G$ has $\alpha(G)\leq 2$, or contains an induced $C_5$ in the neighbourhood of some vertex, or is {\em quasi-line}, meaning every vertex is bisimplicial.

Chv\'atal and Sbihi proved a decomposition theorem for perfect claw-free graphs that yields a polynomial-time recognition algorithm \cite{chvatals88}.  More recently, Shrem, Stern, and Golumbic gave an $O(nm^2)$ algorithm for finding a shortest odd hole in a claw-free graph based on a variant of breadth-first search in an auxiliary graph \cite{shremsg10}.  We solve the same problem, but instead of using local structure we use global structure and take advantage of the similarities between claw-free graphs, quasi-line graphs, and line graphs.  We prove the following:

\begin{theorem}\label{thm:main}
There exists an algorithm that, given a claw-free graph $G$ on $n$ vertices and $m$ edges, finds a smallest odd hole in $G$ or determines that none exists in $O(m^2+n^2\log n)$ time.
\end{theorem}

Fouquet's lemma allows us to focus on quasi-line graphs.  Their global structure, described by Chudnovsky and Seymour \cite{cssurvey}, resembles that of line graphs closely enough that we can reduce the shortest odd hole problem on quasi-line graphs to a set of shortest path problems in underlying multigraphs.  Our algorithm is not much slower than the fastest known recognition algorithms for claw-free graphs:  Alon and Boppana gave an $O(n^{3.5})$ recognition algorithm \cite{alonb87}.  Kloks, Kratsch, and M\"uller gave an $O(m^{1.69})$ recognition algorithm that relies on impractical fast matrix multiplication \cite{klokskm00}.  Their approach takes $O(m^2)$ time using na\"{i}ve matrix multiplication, and more generally $O(m^{(\beta+1)/2})$ time using $O(n^\beta)$ matrix multiplication.

\section{The easy cases: Finding a $C_5$}

We begin by taking advantage of Fouquet's lemma in order to reduce the problem to quasi-line graphs.  We denote the closed neighbourhood of a vertex $v$ by $\tilde N(v)$.

\begin{theorem}\label{thm:alpha2}
Let $G$ be graph with $\alpha(G)\leq 2$.  In $O(m^2)$ time we can find an induced $C_5$ in $G$ or determine that none exists.
\end{theorem}

\begin{proof}
For each edge $uv$ we do the following.  First, we construct sets $X= N(u)\setminus \tilde N(v)$, $Y=N(v)\setminus \tilde N(u)$, and $Z=V(G)\setminus (N(u)\cup N(v))$.  If $u$ and $v$ are in an induced $C_5$ together then all three must be nonempty.  Since $\alpha(G)\leq 2$, we know that both $X$ and $Y$ are complete to $Z$.  Second, we search for $x\in X$ and $y\in Y$ which are nonadjacent -- if such $x$ and $y$ exist then this clearly gives us a $C_5$.  It is easy to see that we can construct the sets in $O(n)$ time, and that we can search for a non-edge between $X$ and $Y$ in $O(m)$ time, since we can terminate once we find one.  Thus it takes $O(m^2)$ time to do this for every edge, and if an induced $C_5$ exists in $G$ we will identify it as $uvyzx$ for any $z\in Z$.
\end{proof}

Kloks, Kratsch, and M\"uller observed that as a consequence of Tur\'an's theorem, every vertex in a claw-free graph has at most $2\sqrt{m}$ neighbours \cite{klokskm00}.  We make repeated use of this fact, starting with a consequence of the previous lemma:

\begin{corollary}
Let $G$ be a claw-free graph with $\alpha(G)\geq 3$.  Then in $O(m^2)$ time we can find an induced $W_5$ in $G$ or determine that $G$ is quasi-line.
\end{corollary}

\begin{proof}
By Fouquet's lemma, any vertex of $G$ is either bisimplicial or contains an induced $C_5$ in its neighbourhood.  For any $v\in V(G)$, we can easily check whether or not $G[N(v)]$ is cobipartite in $O(d(v)^2)$ time.  Since $G$ is claw-free, $d(v)^2 = O(m)$.  Thus in $O(nm)$ time we can determine that $G$ is quasi-line or find a vertex $v$ which is not bisimplicial.

Given this $v$, we can find an induced $C_5$ in $G[N(v)]$ in $O(m^2)$ time by applying the method in the previous proof, since $\alpha(G[N(v)]) \leq 2$.
\end{proof}

Having dealt with these cases made easy by Fouquet's lemma, we can move on to quasi-line graphs with $\alpha\geq 3$, the structure of which we describe now.

\section{The structure of quasi-line graphs}

Given a multigraph $H$ (with loops permitted), its {\em line graph} $L(H)$ is the graph with one vertex for each edge of $H$, in which two vertices are adjacent precisely if their corresponding edges in $H$ share at least one endpoint.  Thus the neighbours of any vertex $v$ in $L(H)$ are covered by two cliques, one for each endpoint of the edge in $H$ corresponding to $v$.  We say that $G$ is a line graph if $G=L(H)$ for some multigraph $H$.

Chudnovsky and Seymour \cite{cssurvey} described exactly how quasi-line graphs generalize line graphs:  a quasi-line graph is essentially either a circular interval graph or can be obtained from a multigraph by replacing each edge with a linear interval graph.

\subsection{Linear and circular interval graphs}

A {\em linear interval graph} is a graph $G=(V,E)$ with a {\em linear interval representation}, which is a point on the real line for each vertex and a set of intervals, such that vertices $u$ and $v$ are adjacent in $G$ precisely if there is an interval containing both corresponding points on the real line.  If $X$ and $Y$ are specified cliques in $G$ consisting of the $|X|$ leftmost and $|Y|$ rightmost vertices (with respect to the real line) of $G$ respectively, we say that $X$ and $Y$ are {\em end-cliques} of $G$.  Given a linear interval representation, if $u$ is to the left of $v$ we say that $u<v$, $u$ is a {\em left neighbour} of $v$, and $v$ is a {\em right neighbour} of $u$.

Accordingly, a {\em circular interval graph} is a graph with a {\em circular interval representation}, i.e.\ $|V|$ points on the unit circle and a set of intervals (arcs) on the unit circle such that two vertices of $G$ are adjacent precisely if some arc contains both corresponding points.  Circular interval graphs are the first of two fundamental types of quasi-line graph.  Deng, Hell, and Huang proved that we can identify and find a representation of a circular or linear interval graph in $O(m)$ time \cite{denghh96}.  We define {\em clockwise neighbours} and {\em counterclockwise neighbours} analogously to left neighbours and right neighbours in linear interval graphs.

\subsection{Compositions of linear interval strips}

We now describe the second fundamental type of quasi-line graph.

A {\em linear interval strip} $(S,X,Y)$ is a linear interval graph $S$ with specified end-cliques $X$ and $Y$.  We compose a set of strips as follows.  We begin with an underlying directed multigraph $H$, and for every every edge $e$ of $H$ we take a linear interval strip $(S_e, X_e, Y_e)$.  For $v\in V(H)$ we define the {\em hub clique} $C_v$ as
$$C_v = \left( \bigcup\{X_e \mid e \textrm{ is an edge out of } v \}\right) \cup \left( \bigcup \{Y_e \mid e \textrm{ is an edge into } v \}\right).$$
We construct $G$ from the disjoint union of $\{ S_e \mid e \in E(H)\}$ by making each $C_v$ a clique; $G$ is then a {\em composition of linear interval strips} (see Figure \ref{fig:strip}).  Let $G_h$ denote the subgraph of $G$ induced on the union of all hub cliques.  That is,
$$G_h = G[\cup_{v\in V(H)} C_v] = G[\cup_{e\in E(H)} (X_e\cup Y_e)].$$

\begin{figure}
\begin{center}
\includegraphics[width=0.7\textwidth]{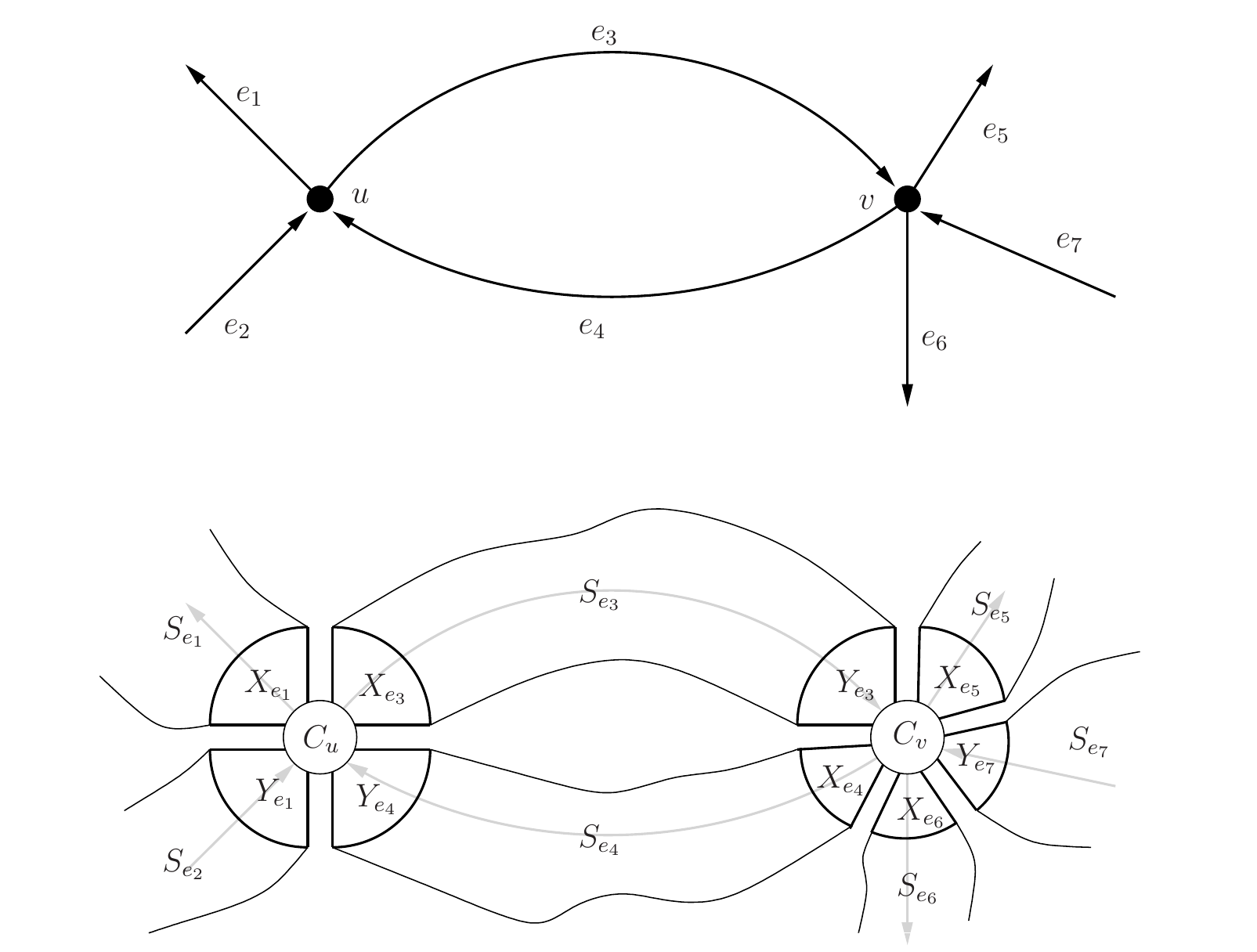}
\end{center}
\caption[A composition of strips]{We compose a set of strips $\{(S_e,X_e,Y_e) \mid e\in E(H)\}$ by joining them together on their end-cliques.  A hub clique $C_u$ will arise for each vertex $u \in V(H)$.}
\label{fig:strip}
\end{figure}

Compositions of linear interval strips generalize line graphs: note that if each $S_e$ satisfies $|S_e|=|X_e|=|Y_e|=1$ then $G = G_h = L(H)$.  

\subsection{Homogeneous pairs of cliques}

A pair of disjoint nonempty cliques $(A,B)$ is a {\em homogeneous pair of cliques} if $|A\cup B|\geq 3$, and every vertex outside $A\cup B$ sees either all or none of $A$ and either all or none of $B$.  These are a special case of {\em homogeneous pairs}, which were introduced by Chv\'atal and Sbihi in the study of perfect graphs \cite{chvatals87}.  It is not hard to show that $G[A\cup B]$ contains an induced copy of $C_4$ precisely if $A\cup B$ does not induce a linear interval graph; in this case we say that $(A,B)$ is a {\em nonlinear} homogeneous pair of cliques\footnote{These were originally called {\em nontrivial} homogeneous pairs of cliques by Chudnovsky and Seymour, who used them in their description of quasi-line graphs \cite{cssurvey}.  We prefer the more descriptive term {\em nonlinear} in part because they are less trivial than {\em skeletal} homogeneous pairs of cliques, which are useful in the study of general claw-free graphs (see \cite{kingthesis}, Chapter 6).}.

\subsection{The structure theorem}

Chudnovsky and Seymour's structure theorem for quasi-line graphs \cite{cssurvey} tells us that all quasi-line graphs are made from the building blocks we just described.

\begin{theorem}
Any quasi-line graph containing no nonlinear homogeneous pair of cliques is either a circular interval graph or a composition of linear interval strips.
\end{theorem}

\section{A proof sketch}

Our approach follows the structure theorem in a straightforward way.  First, we eliminate nonlinear homogeneous pairs of cliques:  Chudnovsky and King give an $O(m^2)$ method for finding an {\em optimal antithickening} of a quasi-line graph \cite{chudnovskyk11}, which leads us to an induced subgraph of $G$ containing no nonlinear homogeneous pair of cliques, and containing a shortest odd hole of $G$.  We explain this in Section \ref{sec:homo}.  Thus we reduce the problem to the cases in which $G$ is a circular interval graph or a composition of linear interval strips.  

We deal with both cases using the same idea.  Given a linear interval strip $(S,X,Y)$, we define a {\em span} of the strip as an induced path with exactly one vertex in each of $X$ and $Y$ (vertices in $X$ and $Y$ must be the endpoints of the path).  To account for parity, for each strip we seek both a shortest span and a {\em near-shortest} span, whose length is greater by one.  Note that there may be no near-shortest span, for example if $S$ is a path.  In Section \ref{sec:strips} we show how to find these paths in $O(m)$ time for a linear interval strip, and explain the simple matter of how to use these paths to find a shortest odd hole in a circular interval graph in $O(m^2)$ time.

To deal with a composition of linear interval strips, we first decompose it using an $O(nm)$ algorithm of Chudnovsky and King \cite{chudnovskyk11}.  This gives us a multigraph $H$ along with linear interval strips $\{(S_e,X_e,Y_e \mid e\in E(H)\}$ such that for each $e$, we have $X_e=Y_e$ or $X_e\cap Y_e=\emptyset$.  We define the {\em span length} $\ell_e$ of a strip $S_e$ as the length of a shortest span, and we let $E_+(H)$ denote the set of edges $e$ of $H$ such that $(S_e,X_e,Y_e)$ contains a near-shortest span (we can determine $E_+(H)$ and find all desired spans in total time $O(nm)$).  The decomposition algorithm in \cite{chudnovskyk11} guarantees that $\ell_e \geq 2$ for $e\in E_+(H)$.

To find a shortest odd hole intersecting $\cup \{V(S_e) \mid e\in E_+(H)\}$, we first assign each edge $e\in E(H)$ weight $\ell_e$.  For each $e\in E_+(H)$, we search for a minimum weight cycle of weight at least four in $H$ containing $e$ by removing $e$ (and all parallel edges of span length 1, if $\ell_e =2$) and using Dijkstra's algorithm to find a shortest path between its endpoints, at a cost of $O(|E(H)|+|V(H)|\log|V(H)|)=O(n\log n)$ per edge $e$.  Since $e\in E_+(H)$, this gives us a shortest odd hole passing through $S_e$: we may change the parity of the hole at a cost of one extra vertex.  As $|E_+(H)|=O(n)$, this step takes $O(n^2\log n)$ time.

We then search for an odd hole not passing through $\cup\{V(S_e) \mid e\in E_+(H)\}$, first discarding every vertex except a shortest span of each strip $S_e$ with $e\notin E_+(H)$.  This leaves a subgraph $G'$, which is actually the line graph of a multigraph $H'$, which we can find in $O(m)$ time.  It remains to find a shortest odd cycle of length $\geq 5$ in $H'$; this will correspond to a shortest odd hole in $G$ not intersecting $\cup\{V(S_e) \mid e\in E_+(H)\}$.  To find such a cycle we first search for a $C_5$ in $O(|V(H)|\cdot|E(H)|)$ time, then search for a longer odd cycle by exploiting restrictions on the chords of a cycle if no $C_5$ exists.  This step takes $O(|V(H)|\cdot|E(H)|+|V(H)|^2\cdot\log |V(H)|) = O(n^2\log n)$ time.  Thus the total running time of the algorithm is $O(m^2+n^2 \log n)$.

\section{Dealing with homogeneous pairs of cliques}\label{sec:homo}

\begin{theorem}
Let $G$ be a quasi-line graph.  Then in $O(m^2)$ time, we can find an induced subgraph $G'$ such that
\begin{itemize}
\item $G'$ is quasi-line and contains no nonlinear homogeneous pair of cliques, and
\item if $G$ contains an odd hole of length $k$, then so does $G'$.
\end{itemize}
\end{theorem}
\begin{proof}This follows easily from results of Chudnovsky and King on {\em optimal antithickenings} of quasi-line trigraphs \cite{chudnovskyk11}.  In terms of graphs, an optimal antithickening of $G$ is a quasi-line graph $G'$ and a matching of edges $M'\subseteq G'$ with the following properties.
\begin{enumerate}
\item There exist disjoint nonlinear homogeneous pairs of cliques $\{(A_i,B_i)\}_{i =1}^k$ such that if we contract each $A_i$ into a vertex $a_i$ and each $B_i$ into a vertex $b_i$, the result is the graph $G'$, and $M'$ is the matching $\{a_ib_i\}_{i=1}^k$.
\item There is no submatching $M''\subseteq M'$ such that $G'-M''$ contains a nonlinear homogeneous pair of cliques.
\end{enumerate}
We construct the graph $G''$ from $G'$ as follows.  For every edge $a_ib_i$ in $M$, we replace $a_i$ with two adjacent vertices $a_i'$ and $a_i''$ such that $N(a_i')=N(a_i)\cup \{a_i''\}$ and $N(a_i'')=N(a_i)\cup \{a_i'\}\setminus \{b_i\}$.  It is straightforward to confirm that $G''$ is an induced subgraph of $G$ (since each $(A_i,B_i)$ is nonlinear), and that $G''$ contains no nonlinear homogeneous pair of cliques (since there is no induced $C_4$ in $G''$ containing both $a_i'$ and $a_i''$ for some $i$).  By the Homogeneous Pair Lemma \cite{chvatals87}, no minimal imperfect graph (and therefore no odd hole) contains a homogeneous pair; it follows easily that no odd hole $C_G$ in $G$ contains more than one vertex from any $A_i$ or $B_i$ for any $i$, and then that $G''$ contains a shortest odd hole of $G$.

We can find $G'$ and $M$ in $O(m^2)$ time \cite{chudnovskyk11}, and given these we can easily construct $G''$ in $O(m^2)$ time.
\end{proof}

\section{Linear interval strips and circular interval graphs}\label{sec:strips}

We now show that we can compute the desired spans of linear interval strips efficiently and use them to find a shortest odd hole in a circular interval graph.  Recall that we can detect and represent linear and circular interval graphs in linear time \cite{denghh96}.

\begin{lemma}
Let $(S,X,Y)$ be a connected linear interval strip with span length $k$, and let $P$ be a shortest span.  If there is a span of length $>k$, there is a span $P'$ of length $k+1$.  Furthermore we can find $P$ and (if it exists) $P'$ in $O(m)$ time.
\end{lemma}

\begin{proof}
We may assume that $X$ and $Y$ are disjoint, otherwise $k=1$ and the lemma is trivial.  We begin by constructing $P = p_1,\ldots, p_k$.  Let $p_1=v_{|X|}$.  For $i=2,\ldots,k$, let $p_i$ be the rightmost neighbour of $p_{i-1}$.  Continue this process until $p_i=p_k$ is in $Y$.  By the structure of a connected linear interval strip, it is easy to see that this gives us a shortest span $P$ in $O(m)$ time.

Our next task is to construct a longest span $Q = q_1,q_2,\ldots,q_\ell$.  We set dummy vertices $v_0$ and $v_{|S|+1}$ with neighbourhoods $X$ and $Y$ respectively, and set $q_0=v_0$ and let $q_{-1}$ be an isolated vertex.  For $i\geq 1$, we let $q_i$ be the leftmost right neighbour of $q_{i-1}$ which is neither adjacent to $q_{i-2}$ nor dominated by $q_{i-1}$.  We continue this process until $q_i=q_\ell$ is in $Y$.  We can clearly do this in $O(m)$ time.  To see that the process results in a span, note that since $S$ is connected, at least one candidate for $q_i$ exists at each step: the rightmost neighbour of $q_{i-1}$ (the addition of $v_{|S|+1}$ ensures that $v_{|S|}$ is not dominated by a vertex outside of $Y$).  To see that $Q$ is a longest span, suppose there exists a longer span $Q'=q'_1,\ldots,q'_{\ell'}$, and let $j$ be the least index such that $q'_j<q_j$.  Then by the construction of $Q$, either $q'_j$ is adjacent to $q_{j-2}$ or $q'_j$ is dominated by $q_{j-1}$.  Either possibility contradicts the fact that $Q$ is an induced path.

If $\ell=k$ then no span of length $k+1$ exists, and if $\ell=k+1$ we are done.  Otherwise, will find some index $i$ such that $P' = q_1,\ldots,q_i,p_i,\ldots,p_k$ is a span.  We simply choose the smallest such $i$ for which $q_i$ sees $p_i$ but not $p_{i+1}$.  We can clearly do this in $O(m)$ time; we now need to prove that this index $i$ exists.

Choose the index $i$ minimum such that $q_{i+1}\leq p_i$ (clearly $i$ exists because $Q$ has length $\geq k+2$, and $i$ must be greater than 1 because $q_2\notin X$).  Then $q_i> p_{i-1}$, therefore $q_i$ sees $p_i$.  Suppose $q_i$ sees $p_{i+1}$; this implies that $q_{i+2}>p_{i+1}$.  However, this contradicts our choice of $p_{i+1}$ as the rightmost neighbour of $p_i$.  Therefore $i$ exists and we have our construction of $P'$.
\end{proof}

\begin{corollary}
Let $(S,X,Y)$ be a connected linear interval strip.  Then in $O(m)$ time we can find the shortest span on an odd number of vertices or determine that none exists.
\end{corollary}

\begin{corollary}\label{cor:cig}
Let $G$ be a circular interval graph with $\alpha(G)\geq 3$.  Then we can find a shortest odd hole in $G$ in $O(m^2)$ time.
\end{corollary}

\begin{proof}
It suffices to show that for any edge $x_ix_j$ of $G$ we can find a shortest odd hole containing $x_ix_j$, or determine that none exists, in $O(m)$ time.  For simplicity, we may assume that $i=1$ and $x_j$ is a clockwise neighbour of $x_i$.  Since $\alpha(G)\geq 3$, it is not a counterclockwise neighbour of $x_i$.

Let $X$ be the set of vertices in $N(x_j)\setminus N(x_1)$, and let $Y$ be the set of vertices in $N(x_1)\setminus N(x_j)$.  Clearly $X$ is a set of clockwise neighbours of $x_j$, and $Y$ is a set of counterclockwise neighbours of $x_1$.  Now let $S = G[V(G)\setminus (\tilde N(x_1)\cap \tilde N(x_2))]$, and observe that $(S,X,Y)$ is a linear interval strip; if it is not connected then there is no hole in $G$ containing $x_1x_j$.  Let $P$ be a shortest span of $(S,X,Y)$ on an odd number of vertices.

If $P$ exists, then $P\cup \{x_1,x_j\}$ induces an odd hole in $G$: $P$ contains at least three vertices, since $X$ and $Y$ are disjoint.  Furthermore observe that for any odd hole $H$ in $G$ containing $x_1$ and $x_j$, $H-\{x_1,x_j\}$ is a span of $S$.  Thus $P$ is a shortest odd hole in $G$ containing $x_1x_j$.  This also tells us that if $P$ does not exist, there is no odd hole in $G$ containing $x_1x_j$.  By the previous corollary, it is clear that we can construct $(S,X,Y)$ and find $P$ in $O(m)$ time.
\end{proof}

\section{Decomposing compositions of linear interval strips}

It now remains to deal with compositions of linear interval strips.  Given a composition of linear interval strips, an {\em optimal strip decomposition} consists of a multigraph $H$ and a set of linear interval strips $\{(S_e,X_e,Y_e \mid e\in E(H)\}$ with the property (among others) that each $S_e$ is either (i) a singleton with $X_e=Y_e=V(S_e)$, or (ii) connected, not a clique, and has $X_e$ and $Y_e$ nonempty and disjoint.

\begin{theorem}[\cite{chudnovskyk11}]\label{thm:stripalgorithm}
Let $G$ be a connected quasi-line graph containing no nonlinear homogeneous pair of cliques.  Then in $O(nm)$ time we can either determine that $G$ is a circular interval graph or find an optimal strip representation of $G$.
\end{theorem}

We remark that $G$ is a line graph precisely if it has a strip decomposition in which every strip $S_e$ is a singleton with $X_e=Y_e=V(S_e)$.

\section{Completing the proof}

We now describe the structure of a hole in relation to an optimal strip decomposition.

\begin{lemma}\label{lem:hole}
Let $G$ be a composition of linear interval strips, and let $C_G$ be a hole in $G$.  Then for any optimal strip decomposition of $G$, the vertex set of $C_G$ can be expressed as 
$$V(C_G)= \cup\{P_e \mid e\in C_H \}$$
where $C_H$ is a cycle (possibly a diad or a loop) in $H$ and $P_e$ is a span of $(S_e,X_e,Y_e)$.
\end{lemma}
\begin{proof}
It is enough to observe two things.  First, $C_G$ must intersect each hub clique $C_v$ exactly 0 or 2 times.  At most twice because $C_v$ is a clique, and not exactly once because if a vertex $u$ of $C_G$ is in $C_v\cap X_e$, and no neighbour of $u$ is in $C_v \cap C_G$, then $C_G$ must contain two nonadjacent vertices in $N(u)\setminus C_v$, a clique.  This is clearly impossible.

Second, for a strip $(S_e,X_e,Y_e)$, if an odd hole intersects $X_e$ it must also intersect $Y_e$, and $S_e\cap C_G$ induces a span of $(S_e,X_e,Y_e)$.  To see this, let $u$ denote the rightmost vertex of $S_e$ in $C_G$.  Clearly $u$ cannot be in $X_e$ unless $X_e=Y_e=\{u\}$, since the strip decomposition is optimal.  Further, the structure of linear interval graphs tells us that at least one neighbour of $u$ must be outside $S_e$, implying that $u\in Y_e$.
\end{proof}

Before dealing with compositions of linear interval strips we must prove a useful lemma on the structure of shortest odd cycles of length $\geq 5$.  For two vertices $u$ and $v$ in a graph, we denote the distance between $u$ and $v$, i.e.\ the length of a shortest $u$-$v$ path, by $\dist(u,v)$.  We denote the length of a path or cycle $P$, i.e.\ the number of edges in it, by $\len(P)$.

\begin{lemma}\label{lem:cycles}
Let $H$ be a graph containing no cycle of length 5, and let $C$ be a shortest odd cycle of length $\geq 7$ in $H$.  Then $C$ contains an edge $v_1v_2$ and an opposite vertex $w$ so that $C$ is the union of $v_1v_2$, a shortest $v_1-w$ path, and a shortest $v_2-w$ path.
\end{lemma}

\begin{proof}
If $C$ has a chord that does not form a triangle with two edges of $C$, then $H$ contains a shorter odd cycle of length $\geq 5$, a contradiction.  Further, $C$ cannot have two non-crossing chords.  Thus there are two consecutive vertices of $C$, call them $v_1$ and $v_2$, such that every chord of $C_H$ has an endpoint in $\{v_1,v_2\}$ and forms a triangle in $H[V(C)]$ containing both $v_1$ and $v_2$.

Suppose $C$ has length $2k+1$ with $k\geq 3$, and let $w$ be the vertex of $C$ opposite from $v_1v_2$.  That is, in the induced subgraph $H[V(C)]$ both $\dist(v_1,w)$ and $\dist(v_2,w)$ are equal to $k$.  It suffices to prove that if in $H$, $\dist(v_1,w)<k$, then $C$ is not a shortest odd cycle of length $\geq 5$ in $H$.  So let $P$ be a shortest $v_1$-$w$ path of length $<k$.  For any two vertices $v'$ and $w'$ in $V(P)\cap V(C)$, let $P_{v'w'}$ be the subpath of $P$ from $v'$ to $w'$ and let $C_{v'w'}$ be the shorter of the two subpaths of $C$ from $v'$ to $w'$.  Since $\len(P)<k$, there must exist $v'$ and $w'$ in $V(P)\cap V(C)$ such that the internal vertices of $P_{v'w'}$ are not in $C$, and $\len(P_{v'w'})<\len(C_{v'w'})$.  Let $C'_{v'w'}$ denote the longer of the two subpaths of $C$ from $v'$ to $w'$.  

Observe that $P_{v'w'}$ cannot be a chord of $C$; this follows from our choice of $v_1$ and $v_2$.  Now $P_{v'w'}$, $C_{v'w'}$, and $C'_{v'w'}$ are internally vertex disjoint $v'-w'$ paths such that $2\leq \len(P_{v'w'}) < \len(C_{v'w'})<\len(C'_{v'w'})$.  It follows that $P_{v'w'}\cup C_{v'w'}$ and $P_{v'w'}\cup C'_{v'w'}$ are both cycles of length between $4$ and $2k$, and one of them is odd.  This contradicts our choice of $C$.
\end{proof}

Along with the results we have already proved, the following lemma immediately implies Theorem \ref{thm:main}.
\begin{lemma}
Given a composition of linear interval strips $G$ containing no nonlinear homogeneous pair of cliques, we can find a shortest odd hole in $G$ in $O(m^2+n^2\log n)$ time.
\end{lemma}
\begin{proof}
We first find an optimal strip decomposition of $G$.  This gives us the underlying multigraph $H$ along with strips $S_e$ for $e\in E(H)$.  Now for each strip $(S_e,X_e,Y_e)$, we find a shortest span $P_e$ and, if one exists, a near-shortest span $Q_e$.  Let $\ell_e$ be the span length of $S_e$.  Set $E_+(H)$ as the set of edges $e$ of $H$ for which $Q_e$ exists.  Let $V_+(G) = \cup\{ V(S_e) \mid e\in E_+(H)  \}$.  We can decompose $G$ determine $E_+(H)$, and find $P_e$ and $Q_e$ in $O(n^2m)$ time.  We find three odd holes (or determine that they do not exist): a shortest odd hole in $S_e$ corresponding to a loop $e$, a shortest odd hole intersecting $V_+(G)$ and not intersecting $\cup \{V(S_e)\mid e\textup{ is a loop}\}$, and a shortest odd hole intersecting neither $V_+(G)$ nor $\cup \{V(S_e)\mid e\textup{ is a loop}\}$.  The shortest of these odd holes is a shortest odd hole in $G$.  It actually suffices to search for the first cycle in $G[\cup \{V(S_e)\mid e\textup{ is a loop}\}]$ and search for the second and third cycles in $G[\cup \{V(S_e)\mid e\textup{ is not loop}\}]$, so we do so.\\

\noindent{\bf Case 1: Holes intersecting} $\cup \{V(S_e)\mid e\textup{ is a loop}\}${\bf .}

If $e$ is a loop in $H$, then $G[V(S_e)]$ is a circular interval graph.  Thus it suffices to find, for each loop $e$, the shortest odd hole in $G[V(S_e)]$.  It follows from Theorem \ref{thm:alpha2} and Corollary \ref{cor:cig} that we can do this in $O(m^2)$ time since the induced subgraphs are disjoint.  Once we have found a shortest odd hole intersecting $\cup \{V(S_e)\mid e\textup{ is a loop}\}$, we can discard the loops of $e$ and their corresponding strips, and henceforth assume that $H$ contains no loops.\\

\noindent{\bf Case 2: Holes intersecting $V_+(G)$.}

Give every edge $e$ of $H$ weight $\ell_e$, and note that as a property of an optimal strip decomposition, every $e\in E_+(H)$ has $X_e$ and $Y_e$ disjoint and therefore $\ell_e\geq 2$.  For each $e\in E_+(H)$ with endpoints $u$ and $v$ we do the following:
\begin{itemize}
\item Begin with $H$ and remove $e$.  If $\ell_e=2$, also remove any edge between $u$ and $v$ with weight 1.
\item Find a minimum weight path $P$ between $u$ and $v$ in the remaining graph.
\end{itemize}

Depending on parity, either 
$$\cup \{P_{e'} \mid e'\in P\}\cup P_e \mathrm{\ \ \   or \ \ \  }\cup \{P_{e'} \mid e'\in P\}\cup Q_e$$ is a shortest odd hole intersecting $S_e$.  It is easy to see that both are holes, and Lemma \ref{lem:hole} tells us that there is no shorter odd hole in $G$ intersecting $S_e$.

We now have, for each $e\in E_+(H)$, the shortest odd hole in $G$ intersecting $S_e$.  Thus we have a shortest odd hole intersecting $V_+(G)$.  Since $|E(H)|\leq n$, and we can run Dijkstra's algorithm in $O(|E(H)|+|V(H)|\log|V(H)|)$ time for each $e\in E_+(H)$ \cite{fredmant87}, we can find the shortest odd hole intersecting $V_+(G)$ in $O(n^2\log n)$ time.\\

\noindent{\bf Case 3: Holes not intersecting $V_+(G)$.}

Let $G'$ be the subgraph of $G$ induced on $\cup \{ P_e \mid e\notin E_+(H) \}$.  Since for $e\notin E_+(H)$, every span of $(S_e,X_e,Y_e)$ has the same length, Lemma \ref{lem:hole} tells us that $G'$ contains a shortest odd hole of $G$ not intersecting $V_+(G)$.  Since $G'$ is a composition of strips, each one of which is an induced path, it is easy to see that $G'$ is a line graph (see Lemma 4.1 in \cite{chudnovskyo07} for a proof of a stronger result).  Let $H'$ be the multigraph of which $G'$ is the line graph -- it is well known that we can find $H'$ in $O(m)$ time (see for example \cite{roussopoulos73}).  Further note that $H'$ has at most $n$ edges, and we are free to remove duplicate edges from $H'$ (they correspond to vertices with the same closed neighbourhood, no two of which can exist in an odd hole).

An odd hole in $G'$ will correspond to an odd cycle (not necessarily induced) in $H'$.  We first search for a $C_5$ in $H'$, knowing that a $5$-hole in $G'$ would be a shortest odd hole in $G$.  We can do so using an $O(|E(H')|\cdot |V(H')|)$-time algorithm of Monien \cite{monien85}, which amounts to $O(n^2)$ time.  So assume $G$ does not contain a $5$-hole.  We therefore seek a shortest odd cycle of length $\geq 5$ in $H'$, which must have length $\geq 7$.  Assume that $H'$ is connected, otherwise we can deal with its connected components individually.

Before we find such a cycle we must compute an all-pairs shortest path matrix using Dijkstra's algorithm.  Normally this involves computing and storing, for each ordered pair of vertices $(u,v)$, $\dist(u,v)$ and a vertex $\next(u,v)$, which is a neighbour of $u$ lying on a shortest $u$-$v$ path.  Instead of storing just one vertex $\next(u,v)$, we want to store a set of up two vertices with this property, if two exist, and otherwise store a single vertex with this property.  This added computation and storage can easily be integrated with the standard $O(n^2\log n)$-time implementation of the all-pairs shortest path version of Dijkstra's algorithm \cite{fredmant87}, and it will soon be clear why we want this extra information.

We wish to find $v_1v_2$ and $w$ representing a shortest odd cycle of length $\geq 7$ in $H'$ as per Lemma \ref{lem:cycles}.  So for each edge $v_1v_2$ we search for a third vertex $w$ with the following properties:
\begin{itemize}
\item $\dist(v_1,w)=\dist(v_2,w)\geq 3$.
\item $|\next(v_1,w) \cup \next(v_2,w)| \geq 2$.
\item For our choice of $v_1v_2$ and subject to the first two requirements, $w$ minimizes $\dist(v_1,w)$.
\end{itemize}
Suppose our choice of $v_1v_2$ minimizes $\dist(v_1,w)$ over all possible choices of $v_1v_2$.  We claim that the union of any internally vertex disjoint shortest $v_1$-$w$ path $P_1$ and shortest $v_2$-$w$ path $P_2$, along with the edge $v_1v_2$, is a shortest odd cycle of length $\geq 7$; we call it $C(v_1,v_2,w)$.  By Lemma \ref{lem:cycles}, it suffices to prove that $P_1$ and $P_2$ actually exist, so choose $P_1$ and $P_2$ so they each intersect a distinct vertex of $\next(v_1,w)\cup \next(v_2,w)$.  Then the paths $P_1$ and $P_2$ must be internally vertex disjoint by our choice of $w$ and the fact that $\len(P_1)=\len(P_2)$.


Therefore to find a shortest odd cycle of length $\geq 7$ in $H'$, we must first find this optimal $v_1v_2$ and $w$.  Given our all-pairs shortest path matrix, this takes $O(|V(H')|\cdot|E(H')|)$ time, i.e.\ constant time for each choice of $v_1v_2$ and $w$.  Once we find $v_1v_2$ and $w$, we can find $C(v_1,v_2,w)$ in $O(|V(H')|)$ time by taking shortest paths from $v_1$ to $w$ and $v_2$ to $w$ that do not intersect at the first vertex.

We have now dealt with every required case in $O(m^2+n^2\log n)$ time, completing the proof.
\end{proof}

\bibliographystyle{plain}

\bibliography{masterbib}
\end{document}